\documentclass{article}
\usepackage{graphicx} 

\usepackage{amsmath, amsfonts, amssymb, amsthm}

\usepackage{CJKutf8}
\usepackage{fullpage}

\usepackage{graphicx}
\graphicspath{{figs/}}
\usepackage{amsmath,amssymb,amsfonts, amsthm}
\usepackage{xcolor}
\usepackage{xspace}
\usepackage{thm-restate}
\usepackage{mathtools}
\usepackage{algorithmicx} 

\usepackage[colorlinks,linkcolor=blue,filecolor=blue,citecolor=blue,urlcolor=blue,pagebackref]{hyperref}

\usepackage{soul}

\usepackage{algorithm}
\usepackage{algpseudocode}
\usepackage{bbold}

\algtext*{EndWhile}
\algtext*{EndIf}
\algtext*{EndFor}
\algrenewcommand\algorithmicrequire{\textbf{Input:}}
\algrenewcommand\algorithmicensure{\textbf{Output:}}

\usepackage[colorinlistoftodos,textsize=tiny,textwidth=2cm,color=red!25!white,obeyFinal]{todonotes}

\usepackage{enumitem}
\setlist{leftmargin=*, topsep=3pt, itemsep=3pt}

\newcommand{\R}{\mathbb{R}}
\newcommand{\Z}{\mathbb{Z}}
\newcommand{\cost}{\mathrm{cost}}

\newcommand{\sfd}{\mathsf{d}}

\newtheorem{theorem}{Theorem}[section]
\newtheorem{lemma}[theorem]{Lemma}
\newtheorem{definition}[theorem]{Definition}

\newtheorem{coro}[theorem]{Corollary}

\DeclareMathOperator*{\E}{{\mathbb{E}}}

\newcommand{\calI}{{\mathcal{I}}}

\newcommand{\opt}{{\mathrm{opt}}}

\allowdisplaybreaks

\title{New Convex Programming Technique for Nash Social Welfare and Scheduling}

\newcommand{\authorspace}{5em}
\author{
    \hspace{\authorspace}
    Yuda Feng\thanks{Nanjing University. \texttt{yudafeng@smail.nju.edu.cn}.}
    \and
    Weijiang Hu\thanks{Nanjing University. \texttt{huweijiang@smail.nju.edu.cn}.}
    \and
    Shi Li\thanks{Nanjing University. \texttt{shili@nju.edu.cn}.}
    \hspace{\authorspace}
}
\date{}

\begin{document}

\maketitle

\begin{abstract}
    We propose a new convex programming relaxation for the weighted Nash social welfare (NSW) problem that achieves a matching \((e^{1/e}\approx 1.445)\)-approximation via the rounding algorithm of Feng and Li. Unlike the exponential-size configuration LP used in prior work, our formulation can be converted into a compact linear program of polynomial size, incurring only an additive loss of \(\ln(1+\epsilon)\) in the objective. This allows the program to be solved directly using standard LP solvers, without the ellipsoid method or dual separation oracles.

    In the unweighted case, we show that our convex program is equivalent to the restricted-spending Fisher market convex program of Cole and Gkatzelis, yielding a constructive proof that its integrality gap is exactly \(e^{1/e}\). With a minor modification, our analysis also gives a simple proof of the \(e^{1/e}\) EF1 gap for the identical agent setting. Finally, we show that our convex programming technique extends to two unrelated machine scheduling problems, recovering the best-known approximation ratios with simpler analyses.
\end{abstract}

\section{Introduction}

We study the problem of allocating a set $M$ of $m$ indivisible items to a set $N$ of $n$ agents so as to maximize the \emph{weighted Nash social welfare (NSW)} of the allocation. Each agent $i\in N$ is associated with a weight $w_i>0$, with $\sum_{i\in N}{w_i} = 1$. Each item $j\in M$ has a value $v_{ij}\in \R_{\geq 0}$ for agent $i\in N$. The goal is to find an allocation $\rho:M\to N$ of items to agents that maximizes the weighted Nash Social Welfare, defined as the weighted geometric mean of agents' valuations:
\begin{align*}
    \prod_{i\in N}{v_i\left(\rho^{-1}(i)\right)^{w_i}},
\end{align*}
where $v_i\left(\rho^{-1}(i)\right) = \sum_{j\in \rho^{-1}(i)}{v_{ij}}$ denotes the value of the bundle allocated to $i$, according to $i$'s valuation. When $w_i = 1/n$ for every agent $i\in N$, we call the problem the unweighted NSW problem. 

Allocating indivisible resources among agents with heterogeneous preferences is a central problem at the intersection of theoretical computer science, game theory, and economics~\cite{daglib/0017737, daglib/0017730, choice/2016, daglib/0017734, daglib/0017738, sp/R2016, daglib/0017736}. Among the many objective functions that have been proposed to capture both efficiency and fairness, the Nash Social Welfare objective stands out for providing a smooth trade-off between utilitarian and egalitarian goals, with applications in bargaining theory~\cite{CM10, LV07, Tho86}, water allocation~\cite{DWL+18, HdLGY14}, and climate agreements~\cite{YvIW17}.

The unweighted NSW problem was already proven to be NP-hard by Nguyen, Nguyen, Roos and Rothe~\cite{aamas/NguyenNRR14}, and APX-hard by Lee~\cite{ipl/Lee17}. Recently, Garg, Hoefer and Mehlhorn~\cite{mor/Garg0M24} improved the hardness of approximation to $\sqrt{8/7} \approx 1.069$.

On the algorithmic side, Cole and Gkatzelis~\cite{siamcomp/ColeG18} introduced a convex program whose optimum solution captures a Fisher market equilibrium with restricted spending on individual items. They designed an efficient rounding algorithm for their convex program with an approximation ratio of $(2e^{1/e} + \epsilon \approx 2.889 + \epsilon)$-approximation.  The approximation ratio was subsequently improved by Cole, Devanur, Gkatzelis, Jain, Mai, Vazirani and Yazdanbod~\cite{sigecom/0001DGJMVY17} to $2$. Independently, Anari, Gharan, Saberi and Singh~\cite{innovations/AnariGSS17} developed an $e$-approximation for the problem using a convex program relaxation based on real stable polynomials. 

The current best-known approximation factor for the problem is $e^{1/e}$ due to Barman, Krishnamurthy and Vaish~\cite{ec/BarmanKV18}. They showed that for an unweighted NSW instance with \emph{identical} agents, 
any EF1 (envy-free up to one item) allocation is $(e^{1/e} \approx 1.445)$-approximate, and then converted the ratio into the approximation ratio for the unweighted NSW problem (where agents are not necessarily identical). We refer to this worst-case approximation factor of EF1 allocations as the EF1 gap.

For the more general weighted NSW problem, Brown, Laddha, Pittu and Singh~\cite{soda/BrownLPS24} presented a $5 \cdot \exp{\left( 2D_{\text{KL}}(w || \frac{\vec{1}}{n}) \right)} = 5 \cdot \exp{(2\log{n} + 2\sum_{i\in N}{w_i \log{w_i}})}$ approximation algorithm, which is a super-constant when the weight vector is far from uniform. It was an open problem to design a constant approximation for weighted NSW. This was solved by Feng and Li~\cite{icalp/FengLi24, theoretics:14642}, who gave an $(e^{1/e} + \epsilon)$-approximation for the problem, matching the best approximation ratio for the unweighted version of the problem. Their approach is based on a natural configuration LP with logarithm of weighted Nash social welfare as the objective, and the Shmoys-Tardos rounding algorithm for the unrelated machine makespan minimization problem \cite{shmoys1993approximation}.

\subsection{Our Results}

Our main contribution is a new convex program for the weighted Nash social welfare (NSW) problem, which also achieves an \(e^{1/e}\)-approximation using the rounding algorithm of Feng and Li \cite{theoretics:14642}. In contrast to the exponential-size configuration LP of \cite{theoretics:14642}, our convex program can be converted into a compact linear program of polynomial size, incurring only an additive loss of \(\ln(1+\epsilon)\) in the objective value. So, it can be solved directly with the loss using standard LP solvers, without resorting to the ellipsoid method and a dual separation oracle.

\begin{theorem}
    \label{thm:rounding}
    The convex program CP\eqref{CP:NSW}, where $f_i$'s are defined as in \eqref{equ:f-NSW}, has an integrality gap of at most $e^{1/e}$. Moreover, there is an efficient rounding algorithm that, given a solution $x$ to CP\eqref{CP:NSW}, outputs an allocation with weighted NSW value at least $e^{-1/e}$ times the value of $x$. 
\end{theorem}

More interestingly, we show that in the unweighted case, our program is equivalent to the convex program CP\eqref{CP:fSR} introduced by Cole and Gkatzelis \cite{sigecom/0001DGJMVY17} that captures the Fisher market equilibrium with restricted spending. The program is defined later in Section~\ref{sec:CP-fSR}. This yields a constructive proof that the integrality gap of this convex program is at most \(e^{1/e}\). Combined with the \(e^{1/e}\) lower bound by \cite{sigecom/0001DGJMVY17}, this shows that the integrality gap of CP\eqref{CP:fSR} is exactly \(e^{1/e}\).

\begin{theorem}
    \label{thm:equivalence}
    For the unweighted Nash social welfare problem, the value of CP\eqref{CP:NSW} is the same as that of CP\eqref{CP:fSR}.
\end{theorem}

\begin{coro}
	The integrality gap of both CP\eqref{CP:NSW} and CP\eqref{CP:fSR} is precisely $e^{1/e}$. 
\end{coro}

With a slight modification of our analysis, we obtain a simple proof of the \(e^{1/e}\) EF1 gap for the unweighted NSW problem with identical agents, which was proved by Barman, Krishnamurthy and Vaish~\cite{ec/BarmanKV18}. Our proof avoids the sequence of adjustment steps. 

Finally, we demonstrate that our convex programming technique naturally extends to two unrelated machine scheduling problems: one with the objective of minimizing the \(L_q\) norm of machine loads, and another with the objective of minimizing weighted completion time when all jobs have the same Smith ratio. The known results \cite{soda/ImLi23, Ola2017unrelated} for both problems use configuration LPs and the rounding algorithm of \cite{shmoys1993approximation}. Similarly, we show that our simple convex programs are sufficient to recover the same approximation ratios as proved in \cite{soda/ImLi23, Ola2017unrelated}. In particular, for the $\frac{1+\sqrt{2}}{2}$ approximation ratio for the weighted completion time problem in \cite{Ola2017unrelated}, we obtain a simpler analysis.

\subsection{Overview of Our Techniques}

As in \cite{theoretics:14642}, we take the logarithm of the weighted Nash social welfare as the objective of our convex program relaxation.  Specifically, we design a concave function $f_i:[0, 1]^M \to \R_{\geq 0}$ for each agent $i \in N$, and aim to maximize $\sum_{i \in N} w_i f_i(x_i)$, where $x_i \in [0, 1]^M$ denotes the fractional allocation of items to agent $i$. 

The strongest choice for $f_i$ is the concave closure of the discrete function mapping $x_i \in \{0, 1\}^M$ to $\ln\sum_jv_{ij}x_{ij}$. This corresponds to the configuration LP introduced by Feng and Li \cite{theoretics:14642}, where $f_i(x_i)$ is defined as the maximum of $\sum_{S \subseteq M}\alpha_S \ln v_i(S)$ over all convex combination decompositions $\sum_{S \subseteq M} \alpha_S \chi^S$ of $x_i$, where $\chi^S \in \{0, 1\}^M$ is the indicator vector for $S$. 

On the opposite side, if all items are assumed to be infinitely divisible, allowing configuration to contain portions of items, then $f_i(x_i)$ is maximized when all configurations have the same value. In this case, the function simplifies to $f_i(x_i) = \ln \sum_{j} v_{ij}x_{ij}$, which preserves the original formula but extends its domain to $[0, 1]^M$. 

Our approach lies between these two extremes. We only treat the largest 1 fractional item of maximum value $v_{ij}$ in $x_i$ as indivisible, while all remaining items are treated as divisible. Under this assumption, the optimal decomposition of items into configurations is given as follows: each configuration contains exactly one indivisible item, and the divisible items are allocated via a water-filling procedure. Then $f_i$ is defined as the average logarithm of the value of these configurations. We show that our $f_i$ is concave by expressing it as the minimum of linear functions.  This is sufficient to recover the $e^{1/e}$ approximation factor using the rounding algorithm of \cite{theoretics:14642}, as in the worst instance, only the largest 1 fractional item is indivisible. 

To show the equivalence of our convex program and that of \cite{sigecom/0001DGJMVY17} via market equilibrium, we establish a correspondence between whether items that are covered by water in the water filling procedure defining $f_i$'s, and items whose market prices are at most 1. Thanks to the uniform weight vector, this classification of items is the same for all agents. This uniformity is why the convex program of \cite{sigecom/0001DGJMVY17} applies only to the unweighted case.

In both the analysis of our rounding algorithm and our simplified proof of the EF1 gap, we define a liquidization operation that converts an item into arbitrarily small items with the same total value, this operation preserves total values but allows us to treat the liquidized items as divisible in the analysis. In the EF1-gap proof, we focus on the agent receiving the smallest value, say $\psi$. We then liquidize items of total value exactly  $\psi$ for every agent, leaving only a portion of the largest item unliquidized. This transformation can only increase the gap. In the optimal solution, the liquidized items are allocated using the water-filling procedure. The remaining task of analyzing the gap between our solution and the optimal one is straightforward.

\subsection{Other Related Work}
A more general function family that has been studied for the NSW problem is submodular functions. In contrast to the additive valuation function setting, each agent $i\in N$ is associated with a monotone submodular valuation function $v_i : 2^M \to \R_{\geq 0}$, and the goal remains to find an allocation $\rho:M\to N$ that maximizes $\prod_{i\in N}{v_i\left(\rho^{-1}(i)\right)^{w_i}}$. As in the additive case, we have the unweighted submodular NSW problem, where all agents have weights $1/n$, and the weighted submodular NSW problem, where the weight vector $w$ can be any vector in $[0, 1]^N$ satisfying $\sum_{i \in N}w_i = 1$. 

In the unweighted case, the current best-known inapproximability for submodular NSW is $e/(e-1) \approx 1.582$, proved by Garg, Kulkarni and Kulkarni~\cite{talg/GargKK23}, even when the number of agents is a constant.

Algorithmically, the first $O(1)$-approximation for unweighted submodular NSW was obtained by Li and Vondrák~\cite{soda/LiV21} via convex programming, with an approximation ratio of $e^3 / (e-1)^2$. Subsequently, Garg, Husic, Li, Végh and Vondrák~\cite{stoc/GargHLVV23} introduced an elegant local-search method that improved the approximation ratio to $4$ for the unweighted case. The same framework also yields an $O(nw_{\max})$-approximation for the weighted case~\cite{stoc/GargHLVV23}, where $w_{\max} = \max_{i\in N}{w_i}$. 

Whether one can obtain an $O(1)$-approximation remained open. This was resolved by Feng, Hu, Li and Zhang~\cite{stoc/FHLZ25}, who developed a $233$-approximation based on the configuration LP from \cite{theoretics:14642}, giving the first constant-factor approximation for the weighted case. Very recently, Bei, Feng, Hu, Li and Zhang~\cite{BFHLZ25} significantly advanced this line of work by presenting a $3.56$-approximation, simultaneously improving the best-known ratios for both weighted and unweighted submodular NSW. They solve the same configuration LP via a stronger separation oracle that loses an $e/(e-1)$ factor only on small items, and then round the solution using a new bipartite multigraph construction to achieve the improved constant. 

Throughout, we shall simply refer to the NSW problems with additive valuations as unweighted/weighted NSW problems, as they are our focuses of the paper. 
\smallskip


For the unrelated machine scheduling problem with the objective of minimizing the $L_q$ norm of machine loads ($q \geq 1$), Im and Li \cite{soda/ImLi23} obtained an approximation ratio of $O_q(1)$-approximation algorithm, using a time-index LP relaxation and the Shmoys-Tardos rounding algorithm \cite{shmoys1993approximation}. This framework improves upon some earlier results \cite{awerbuch1995load, azar2005convex, kumar2009unified}. In particular, their approximation ratio for $q = 2$ is $\sqrt{4/3} \leq 1.155$.  

The unrelated machine weighted completion time problem has been extensively studied in the literature. Following a sequence of work~\cite{Har24, soda/ImLi23, IS20}, Li \cite{soda/000125} gave the current best approximation ratio of $1.36 + \epsilon$ for the problem. 
The case we are interested is when all jobs have the same Smith ratios across all machines, 
for which Kalaitzis, Svensson and Tarnawski \cite{Ola2017unrelated} gave a $\frac{1+\sqrt{2}}{2}$-approximation. This is also based on the Shmoys-Tardos rounding algorithm, along with the configuration LP relaxation for the problem. 

\paragraph{Organizations}
The rest of the paper is organized as follows. In Section~\ref{sec:prelim}, we introduce the problem definitions, notation and other preliminaries for NSW problem and the related scheduling problems. In Section~\ref{sec:CP}, we present our main convex program, and show the equivalence of CP\eqref{CP:NSW} and CP\eqref{CP:fSR} in unweighted case, proving Theorem~\ref{thm:equivalence}. In Section~\ref{sec:rounding}, we analyze the rounding algorithm of CP\eqref{CP:NSW} and prove Theorem~\ref{thm:rounding}. In Section~\ref{sec:EF1-gap}, we provide a simple proof of the $e^{1/e}$ EF1 gap. Finally Section~\ref{sec:scheduling}, we apply our convex programming technique to two scheduling problems.
\section{Preliminaries}
\label{sec:prelim}

\subsection{Problem Definitions}
\label{sec:problem-definitions}

For the weighted additive Nash Social Welfare (NSW) problem, it is convenient for us to avoid agent-item pairs $ij$ with $0$ value. So in our definition of the problem, we are given a set $N$ of $n$ agents, each $i \in N$ with a weight $w_i \in (0, 1]$ subject to $\sum_{i \in N}w_i = 1$, a set $M$ of $m$ items, a set $E \subseteq N \times M$ of edges, along with a value $v_{ij} \in \R_{>0}$ for every $ij \in E$.  We assume every $i \in N$ or $j \in M$ is incident to at least one edge in $E$. Let $M_i := \{j \in M: ij \in E\}$  for every $i \in N$ and $N_j = \{i \in N: ij \in E\}$ for every $j \in M$.  Our goal is to find an allocation $\rho:M \to N$ such  that $\rho(j)j \in E$ for every $j \in M$, so as to maximize $\prod_{i \in N}v_i(\rho^{-1}(i))^{w_i}$, where we use $v_i(M') := \sum_{j \in M'} v_{ij}$ for every $i \in N$ and $M' \subseteq M_i$.  In the unweighted Nash Social Welfare problem, we have $w_i = \frac1n$ for every $i \in N$.

We shall consider the following general unrelated machine scheduling problem. There is a set $M$ of $m$ machines, a set $J$ of $n$ jobs, a size $p_{ij} \in \R_{> 0}$ for every $i \in M$ and $j \in J$. We are given an increasing convex function $\theta:\R_{\geq 0} \to \R_{\geq 0}$ with $\theta(0) = 0$. Our goal is to find an allocation $\rho:J \to M$ so as to minimize $\sum_{i \in M} \theta(\text{load}_i)$, where $\text{load}_i := q_i(\rho^{-1}(i))$ and $q_i(J'):=\sum_{j \in J'} p_{ij}$ for every $J' \subseteq J$. We call the problem the unrelated machine scheduling problem with $\sum_{i \in M} \theta(\text{load}_i)$ objective.\footnote{We remark the notation change from NSW problems to scheduling problems. In both settings we need to allocate some objects to some players. $M$ denotes the objects (items) in NSW problems, and players (machines) in scheduling problems. We always have $m = |M|$. $n$ is the number of players ($n = |N|$) in NSW problems, and number of objects ($n = |J|$) in scheduling problems.}  When $\theta(x) = x^q$ for $q \geq 1$, the objective becomes the $L_q$ norm, after taking the $q$-th root of the objective. 

In the unrelated machine scheduling with weighted completion time objective and uniform Smith ratios, we are given $M, m, J, n$ and $(p_{ij} \in \R_{>0})_{i \in M, j \in J}$ as in the previous problem. Our goal is to find an allocation $\rho:J \to M$ so as to minimize $\frac12\sum_{i \in M} \left(p_i(\rho^{-1}(i))^2 + \sum_{j \in \rho^{-1}(i)}p_{ij}^2\right)$. When a job $j$ is allocated to $i = \rho(j)$, it has processing time $p_{ij}$ and weight $p_{ij}$. Namely, the Smith ratio of any job on any machine is always $1$. On every machine $i \in M$, the weighted completion time of scheduling $\rho^{-1}(i)$ on $i$ is $\frac12 \left(p_i(\rho^{-1}(i))^2 + \sum_{j \in \rho^{-1}(i)}p_{ij}^2\right)$. As all jobs have the same Smith ratio, any permutation of jobs is optimum.  The objective is closely related to the $L_2$ norm of machine loads.


\subsection{EF1 Allocation and EF1-Gap for Unweighted Additive NSW with Identical Agents} 
Suppose we are given an unweighted NSW instance with $n$ \emph{identical} agents $N$, and $m$ indivisible items $M$. As the valuation functions are identical, we simply use $v_j$ to denote the value of item $j$ for any agent $i \in N$.
\begin{definition}
    An allocation $\sigma: M \to N$ of items to agents is said to be envy-free up to 1 item (EF1) if for every two agents $i, i’ \in N$, we have $v(\sigma^{-1}(i)) \geq v(\sigma^{-1}(i’)) - \max\{v_j: j \in \sigma^{-1}(i’)\}$. 
\end{definition}
So in an EF1 allocation $\sigma$, for any pair $i, i'$ of agents, $i$ does not envy $i’$ once we remove the largest item in the bundle of agent $i’$.  \cite{ec/BarmanKV18} proved the following theorem:
\begin{theorem}[\cite{ec/BarmanKV18}]
    \label{thm:EF1-gap}
Consider an unweighted NSW instance with identical agents, defined by $n, N, m, M$ and $v$. Let $\sigma$ be an EF1 allocation, and $\opt$ be the optimum NSW value of the instance. Then, we have 
\begin{align*}
	\left(\prod_{i \in N} v_i(\sigma^{-1}(i))\right)^{1/n} \geq e^{-1/e} \cdot \opt.
\end{align*}
\end{theorem}
The $e^{1/e}$ factor is called the EF1 gap of for the NSW problem. At a high level, the proof in \cite{ec/BarmanKV18} compares an arbitrary EF1 allocation to an optimal partially-fractional allocation via a sequence of swap operations, derives a lower bound for the EF1 allocation and an upper bound on the optimal NSW, and then optimizes these bounds to obtain the $e^{-1/e}$ guarantee. However, the argument relies on several structural lemmas and nontrivial inequalities, which makes the proof rather complicated.


\subsection{Convex Program CP\eqref{CP:fSR} of \cite{sigecom/0001DGJMVY17} for Unweighted NSW}
\label{sec:CP-fSR}
    \cite{sigecom/0001DGJMVY17} introduced a market equilibria based convex program for the unweighted additive Nash social welfare problem. The convex program, denoted as CP\eqref{CP:fSR}, is defined as follows:
    \begin{align}
        \max \qquad 
        \frac1n\left(\sum_{ij \in E}b_{ij} \ln v_{ij} - \sum_{j \in M} q_j \ln q_j\right)
        \tag{\text{f-SR}} \label{CP:fSR}
    \end{align}\vspace*{-10pt}
    \begin{align*}
        \sum_{i \in N_j} b_{ij} = q_j, \forall j \in M; \quad 
        \sum_{j \in M_i} b_{ij} = 1, \forall i \in N; \quad
        b_{ij} \geq 0, \forall ij \in E; \quad
        q_j \leq 1, \forall j \in M.
    \end{align*}

    We note that \cite{siamcomp/ColeG18} uses the NSW value as the objective, while we instead use its logarithm, for the ease of comparison with our convex program. Consider an allocation $\rho:M \to N$ for the instance. Then, we define $b_{ij} = \frac{v_{ij}}{v_i(\rho^{-1}(i))}$ if $i = \rho(j)$ and $b_{ij} = 0$ otherwise. $q_j = b_{\rho(j)j}$. 
    This gives us a valid solution to CP\eqref{CP:fSR}. The logarithm of the NSW of $\rho$ is
    \begin{align*}
        \frac1n \sum_{i \in N}\ln v(\rho^{-1}(i)) 
        &=\frac1n\sum_{i \in N} \sum_{j \in M_i} b_{ij} \ln \frac{v_{ij}}{q_j} = \frac1n\left(\sum_{ij \in E} b_{ij} \ln v_{ij} - \sum_{ij \in E}b_{ij}\ln q_j\right)\\
        &= \frac1n\left(\sum_{ij \in E} b_{ij} \ln v_{ij} - \sum_{j \in M}\ln q_j\right),
    \end{align*}
    which is exactly the objective of CP\eqref{CP:fSR}. \smallskip

    \cite{siamcomp/ColeG18} proved the following lemma:    
    \begin{lemma}
        \label{lemma:pricing}
        With a scaling of the valuation functions $v_i, i \in N$, there exists a price vector $p \in \R_{>0}^{M}$ such that the following conditions hold for the optimum solution $(b, q)$ to CP\eqref{CP:fSR}:
        \begin{itemize}
            \item $q_j = \min\{p_j, 1\}$ for every $j \in M$.
            \item $v_{ij} \leq p_j$ for every $ij \in E$.
            \item If $v_{ij} < p_j$ for some $ij \in E$, then $b_{ij} = 0$.
        \end{itemize}
    \end{lemma}
    The lemma says that one can define prices $p_j$ for items so that the value of an item $j$ to any agent is at most its price. So the \emph{bang-per-buck (BaB)} for any $ij$ pair is at most $1$. Moreover, in the solution $(b, q)$, $i$ will only be allocated items with BaB exactly 1. When agent $i$ gets a value $b_{ij}$ from item $j$, he pays $\$b_{ij}$ to $j$. Every agent $i$ spends $\$1$, and the money spent on any item $j$ is $\$q_j := \min\{\$p_j, \$1\}$. This achieves a market equilibrium with ``restricted spending''.

\subsection{Rounding Algorithm of \cite{theoretics:14642}}
\label{sec:rounding-framework}
We revisit the rounding algorithm of \cite{theoretics:14642}, which is the same as that of \cite{shmoys1993approximation} for the unrelated machine makespan minimization problem. This is the same rounding algorithm used in both \cite{soda/ImLi23} and \cite{Ola2017unrelated}.

We shall describe the rounding algorithm in the context of weighted Nash Social Welfare. We are given a fractional allocation $x \in [0, 1]^E$ of items to agents: we have $\sum_{i \in N_j} x_{ij} = 1$ for every $j \in M$. We further have $\sum_{j \in M_i} x_{ij} \geq 1$ for every $i \in N$ since otherwise the value of the solution $x$ will be $0$. (This property does not hold for scheduling problems, and it is not required by the rounding algorithm.) Throughout, for the given $x$, we shall use $x_i := (x_{ij})_{j \in M_i}$ to denote the fractional allocation for agent $i$.

For each agent $i \in N$, we sort items in $M_i$ by non-increasing $v_{ij}$ and partition the fractional items $x_i$ into $q_i := \lceil \sum_{j\in M_i} x_{ij} \rceil$ groups $G_i = \{ z^{i,1}, z^{i,2}, \dots, z^{i, q_i} \in [0, 1]^{M_i} \}$ using this order, each containing $1$ fractional item except for the last one. Here each $z^{i,t}$ is a fractional bundle over $M_i$ with $|z^{i,t}|_1 \leq 1$, and for each $j\in M_i$, $z^{i,t}_j$ denotes the fraction of item $j$ contained in $z^{i,t}_j$. 
More formally, we define a total order $\succ_i$ over $M_i$ that sorts the items in non-increasing order of $v_{ij}$ with ties broken arbitrarily. Thus, $j \succ_i j'$ implies $v_{ij} \geq v_{ij'}$. Then $z^{i,1}, z^{i,2}, \cdots, z^{i, q_i}$ are vectors in $[0, 1]^{M_i}$ satisfying the following properties:
\begin{itemize}[topsep=5pt,leftmargin=*]
    \item $\sum_{t = 1}^{q_i}z^{i, t} = x_i$. 
    \item $|z^{i,t}|_1 = 1$ for every $t = 1, 2, \cdots, q_i - 1$. 
    \item For any two integers $t, t'$ with $1 \leq t < t' \leq q_i$ and two items $j, j' \in M_i$ with $j \succ_i j'$, it can not happen that $z^{i,t}_{j'} > 0$ and $z^{i, t'}_j > 0$.
\end{itemize}
Notice that under the above 3 conditions, the choice of $(z^{i, 1}, z^{i, 2}, \cdots, z^{i, q_i})$ is unique.

This yields a collection $G = \biguplus_i G_i$ of groups and a fractional matching between the groups $G$ and items $M$: every group $z \in G$ is matched to an item $j$ with fraction $z_j$ (or $0$ if $j$ is not in the domain of $z$). So each group $z$ is matched to an extent of $|z|_1$, and each item is matched to an extent of $1$. By the way we construct the groups, a group which is not the last group for an agent $i$ is matched to an extent of $1$. 

In the randomized rounding algorithm, we arbitrarily partition the fractional matching into a convex combination of (partial-)matchings between $G$ and $M$ and randomly pick a matching from the combination, according to their masses.  Each integral matching allocates each item to one group, hence induces an integral allocation of items to agents: an item $j \in M$ is allocated to the agent $i$ if it is matched to a group in $G_i$. Marginal probabilities are maintained by the rounding algorithm: the probability that an item $j$ is allocated to an agent $i$ is precisely $x_{ij}$ for any $ij \in E$. 

The algorithm can be derandomized by using a polynomial-sized decomposition and outputting the best matching from the decomposition. However, as is typical, it is more convenient to analyze the randomized version of the rounding algorithm. 



\subsection{Creating Unweighted NSW Instance $\calI$ with Identical Agents and EF1 Allocation $\sigma$ for a Fixed Agent $i$} 
\label{sec:creating-identical}
Let $\rho$ be the random allocation given by the algorithm.  In the analysis of the randomized rounding algorithm, we shall focus on a fixed agent $i$ most of the time. We create an unweighted NSW instance $\calI$ with identical agents, which are all copies of the agent $i$, and an EF1 allocation $\sigma$ for $\calI$.  

Let $\Delta$ be a sufficiently large integer such that $\Pr[\rho^{-1}(i) = S]$ for every $S$ is an integer multiple of $1/\Delta$, so is $x_{ij}$ for every $j \in M_i$. We create $\Delta$ copies of the agent $i$, and index them as $[\Delta]$. We create $\Delta x_{ij}$ copies of item $j$ for every $j \in M_i$. The allocation $\sigma$ is defined according to the output distribution of the rounding algorithm: for every $S \subseteq M_i$, there are precisely $\Delta \cdot \Pr[\rho^{-1}(i) = S]$ agents in $[\Delta]$ who gets a copy of $S$, in the allocation $\sigma$. It was shown in \cite{theoretics:14642} that $\sigma$ is an EF1-allocation. 

The procedure of creating $\calI$ and $\sigma$ also applies to the rounding algorithm for scheduling problems.

\section{The Convex Program}
\label{sec:CP}
    In this section, we describe our convex program which yields an $e^{1/e}$-approximation for weighted NSW. We prove that it is equivalent to CP\eqref{CP:fSR} in the unweighted case, showing that the integrality gap of CP\eqref{CP:fSR} is precisely $e^{1/e}$. 
    
    For every $i \in N$, and $x_i \in [0, 1]^{M_i}$ satisfying $|x_i|_1 \geq 1$, we define $h_i(x_i)$ to be any real $h > 0$ satisfying 
    \begin{align*}
        \sum_{j \in M_i} \min\{v_{ij}, h\} x_{ij} = h.    
    \end{align*}
    If $|x_i|_1 > 1$, the choice of $h_i(x_i)$ is unique, which can be determined by the following procedure. Let $x'_i \in [0, 1]^{M_i}$ be the vector that maximizes $\sum_{j \in M_i}v_{ij}x'_{ij}$ subject to $x'_i \leq x_i, |x'_i| = 1$. In other words, $x'_i$ consists of the largest one fractional item in $x_i$. We construct a histogram for $x'_i$, where every item $j \in M_i$ is represented by a rectangle of width $x'_{ij}$ and height $v_{ij}$. Let $V := \sum_{j \in M_i}(x_{ij} - x'_{ij})v_{ij}$ denote the total value of remaining fractional items. Treating this value as $V$ units of water, we pour the water into the histogram from bottom to top. Then $h_i(x_i)$ is the resulting water level. See Figure~\ref{fig:f} for illustration of the definition. When $|x_i|_1 = 1$, no water is poured, and $h_i(x_i)$ can be any real number in $\big(0, \min\{v_{ij}: j \in M_i, x_{ij}> 0\}\big]$. \medskip 

    \begin{figure}
        \centering
        \includegraphics[width=0.35\textwidth]{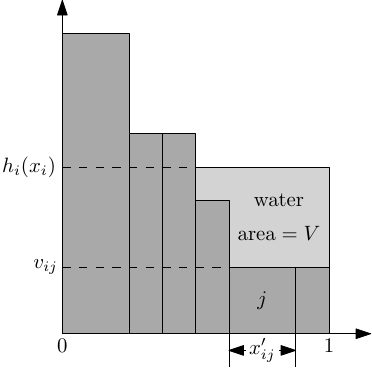}
        \caption{Definition of $h_i(x_i)$ and $f_i(x_i)$ when $|x_i|_1 > 1$. The dark gray part denotes the 1 fractional largest items in $x_i$. The light gray part denotes the $V$ units of water. $f_i(x_i)$ is the average height of the histogram on a logarithmic scale. }
        \label{fig:f}
    \end{figure}
    
    With $h_i(x_i)$ defined, we define $f_i(x_i)$ as follows:
    \begin{align}
        f_i(x_i):=\sum_{j\in M_i: v_{ij} > {h_i(x_i)}} x_{ij} \ln\frac{v_{ij}}{h_i(x_i)} + \ln h_i(x_i). \label{equ:f-NSW}
    \end{align}
    So, $f_i(x_i)$ is the average height of the histogram in Figure~\ref{fig:f} on a logarithmic scale. Notice that when $|x_i| = 1$, we have $f_i(x_i) = \sum_{j \in M_i}x_{ij} \ln v_{ij}$, which is independent of the choice of $h_i(x_i)$. \medskip

    We give an alternative description of $f_i(x_i)$ which will prove the concavity of the function.
    For every $i \in N, x_i \in [0, 1]^{M_i}$ with $|x_i| \geq 1$ and $h > 0$, we define
    \begin{align*}
        g_i(x_i, h):= \sum_{j \in M_i: v_{ij} > h} x_{ij} \ln \frac{v_{ij}}{h} + \ln h  + \frac 1h\sum_{j \in M_i} \min\{v_{ij}, h\}x_{ij} - 1. 
    \end{align*}
    Notice that for a fixed $h$, the function $g_i$ is linear in $x_i$.

    \begin{lemma}
        \label{lemma:fx}
        For any $i \in N$, $x_i \in [0, 1]^{M_i}$ satisfying $|x_i|_1 \geq 1$, we have
        \begin{align*}
            f_i(x_i)=\min_{h > 0} g_i(x_i, h) = g_i(x_i, h_i(x_i)).
        \end{align*}
    \end{lemma}

    \begin{proof}
        We calculate $\frac{\partial g_i}{\partial h}$ for any $h \neq v_{ij}$ for any $j \in M_i$:
        \begin{align*}
            \frac{\partial g_i}{\partial h} &= \sum_{j \in M_i: v_{ij}>h} (-\frac{x_{ij}}{h}) + \frac1h - \sum_{j \in M_i: v_{ij} < h} \frac{v_{ij}x_{ij}}{h^2} = \frac{1}{h} - \frac1{h^2}\sum_{j \in M_i}x_{ij} \min\{v_{ij}, h\} \\
            &= \frac{1}{h^2}\left(h - \sum_{j \in M_i}x_{ij} \min\{v_{ij}, h\}\right).
        \end{align*}
        So, the derivative is continuous on $h > 0$.  Moreover, $\frac{\partial g_i}{\partial h} \leq 0$ when $h < h_i(x_i)$ and $\frac{\partial g_i}{\partial h} \geq 0$ when $h > h_i(x_i)$. So, the function is minimized when $h = h_i(x_i)$.  Moreover, when $h = h_i(x_i)$, we have $g_i(x_i, h) = f_i(x_i)$ as $\sum_{j \in M_i}\min \{v_{ij}, h_i(x_i)\}x_{ij} = h_i(x_i)$.
    \end{proof}

    As $g_i(x_i, h)$ is a linear function of $x_i$ for any fixed $h > 0$, we have
    \begin{coro}
        $f_i$ is concave over its domain. 
    \end{coro}

    Our convex program can simply be written as follows:
    \begin{align}
        \max \qquad \sum_{i \in N} w_i f_i(x_i) \label{CP:NSW}
    \end{align} \vspace*{-5pt}
    \begin{align*}
        \sum_{j \in M_i} x_{ij} \geq 1, \forall i \in N; \qquad
        \sum_{i \in N_j} x_{ij} = 1,\forall j \in M; \qquad
        x_{ij} \geq 0, \forall ij \in E. 
    \end{align*}

    \subsection{Solving CP\eqref{CP:NSW} via Discretization}
    Unlike \cite{sigecom/0001DGJMVY17}, we could not show the rationality of the optimum solution of the convex program. So, it is not clear if the convex program can be solved in polynomial time, and more specifically, if the optimum solution can be represented using rational numbers. Nevertheless, the standard discretization technique allows us to solve the CP approximately within an additive error of $\ln(1+\epsilon)$ for any constant $\epsilon > 0$.  For every $i \in N$, we let $\ell_i := \min\{v_{ij}: j \in M_i\}$ and $r_i := \sum_{j \in M_i}v_{ij}$, so that we always have $h_i(x_i) \in [\ell_i, r_i]$. (In case $|x_i| = 1$, we can choose $h_i(x_i) = \min\{v_{ij}: j \in M_i, x_{ij} > 0\}$.) Then, let $H_i = \{r_i (1+\epsilon)^{-t} \geq \ell_i: t \in \Z_{\geq 0} \}$. We shall use 
    \begin{align*}
        \bar{f}_i(x_i):= \min_{h \in H_i} g_i(x_i, h)
    \end{align*}
     as an approximation of $f_i(x_i)$.  
    \begin{lemma}
        $f_i(x_i) \leq \bar{f}_i(x_i) < f_i(x_i) + \ln(1+\epsilon)$ for every $x_i \in [0, 1]^{M_i}$ with $|x_i|_1 \geq 1$. 
    \end{lemma}
    \begin{proof}
        $f_i(x_i) \leq \bar{f}_i(x_i)$ holds trivially and so we only need to prove the second inequality. Let $\bar h$ be the smallest number in $H_i$ with $\bar h \geq h_i(x_i)$. Then, we have $h_i(x_i) \leq \bar h < (1+\epsilon)h_i(x_i)$. From the proof of Lemma~\ref{lemma:fx}, we know $\frac{\partial g_i}{\partial h} \leq \frac{1}{h}$.  So, we have $g_i(x_i, \bar h) - g_i(x_i, h_i(x_i)) \leq \int_{h_i(x_i)}^{\bar h} \frac{1}{h} \sfd h = \ln\frac{\bar h}{h_i(x_i)} < \ln(1+\epsilon)$. The lemma follows from that $\bar{f}_i(x_i) \leq g_i(x_i, \bar h)$ and $f_i(x_i) = g_i(x_i, h_i(x_i))$.
    \end{proof}

    Therefore, we can consider the new convex program where the objective is to maximize $\sum_{i \in N}w_i \bar{f}_i(x_i)$.  The new convex program can be reformulated explicitly as a polynomial-sized linear program: we change our objective function to $\sum_{i \in N} w_i \bar{f}_i$, and add linear constraints $\bar{f}_i \leq g_i(x_i, h)$ for every $i \in N$ and $h \in H_i$ to the program.

    \subsection{Equivalence of CP\eqref{CP:NSW} and CP\eqref{CP:fSR} in Unweighted Case}
     In this section, we prove Theorem~\ref{thm:equivalence} by showing the equivalence of CP\eqref{CP:NSW} and CP\eqref{CP:fSR} for unweighted NSW. 

    \begin{proof}[Proof of Theorem~\ref{thm:equivalence}]
        First, we show the value of CP\eqref{CP:NSW} is at most that of CP\eqref{CP:fSR}. Given a valid solution $x \in [0, 1]^E$ to CP\eqref{CP:NSW}, we convert it to a valid solution to CP\eqref{CP:fSR}, with value at least that of $x$ to CP\eqref{CP:NSW}.

        

        By scaling the valuation functions, we assume $h_i(x_i) = 1$ for every $i \in N$. We construct a solution to CP\eqref{CP:fSR} as follows: 
        \begin{align*}
            b_{ij} := x_{ij} \min\{v_{ij}, 1\}, \forall ij \in E, \quad \text{and} \quad q_j := \sum_{i \in N_j} b_{ij}, \forall j \in M.
        \end{align*}
        For every $j \in M$, we have $q_j \leq \sum_{i \in N_j} x_{ij} = 1$.  For every $i \in N$,  we have $\sum_{j \in M_i} b_{ij} = \sum_{j \in M_i} x_{ij} \min\{v_{ij}, 1\} = 1$,  by the definition of $h_i(x_i)$ and that $h_i(x_i) = 1$. Therefore, the solution $(b, q)$ to CP\eqref{CP:fSR} is valid.
        
        We now compute the value of $(b, q)$ to CP\eqref{CP:fSR}. 
        Observe that for every $j \in M$, $\sum_{i \in N_j} \frac{b_{ij}}{q_j} = 1$ and $\sum_{i \in N_j}\frac{b_{ij}}{q_j} \cdot \frac{1}{\min\{v_{ij}, 1\}} = \sum_{i \in N_j}\frac{x_{ij}}{q_j} = \frac{1}{q_j}$. By the concavity of the logarithm function $\ln(\cdot)$, we obtain $\sum_{i \in N_j} \frac{b_{ij}}{q_j} \cdot \ln \frac{1}{\min\{v_{ij}, 1\}} \leq \ln \frac{1}{q_j}$,
        which is equivalent to
        \begin{align*}
            \sum_{i \in N_j} b_{ij} \ln \min\{v_{ij},1\} \geq q_j \ln q_j.
        \end{align*}

        So, we have
        \begin{align*}
            \sum_{ij \in E} b_{ij} \ln v_{ij} &=
            \sum_{ij \in E} b_{ij} \ln \min\{v_{ij}, 1\} +
            \sum_{ij \in E : v_{ij} > 1} b_{ij} \ln v_{ij}\\
            &\geq \sum_{j \in M} q_j\ln q_j + \sum_{ij \in E:v_{ij}>1} x_{ij}\ln v_{ij} = \sum_{j \in M} q_j\ln q_j + \sum_{i \in N} f_i(x_i),
        \end{align*}
        by the definition of $f_i(x_i)$ and that $h_i(x_i) = 1$. Therefore,
        \begin{align*}
            \frac1n\left(\sum_{ij \in E} b_{ij} \ln v_{ij} - \sum_{j \in M} q_j \ln q_j\right)
            \geq \frac1n\sum_{i \in N} f_i(x_i). 
        \end{align*}\medskip

        Now we prove that the value of CP\eqref{CP:fSR} is at most that of CP\eqref{CP:NSW}. Let $(q, b)$ be the optimum solution to CP\eqref{CP:fSR}. We scale the valuation functions as stated in Lemma~\ref{lemma:pricing}, and let $p \in \R_{>0}^M$ be the price vector from the lemma. 
        We construct a solution $x \in [0, 1]^E$ to CP\eqref{CP:NSW} whose value is at least that of $(q, b)$. The definition of $x$ is simple:
        \begin{align*}
            x_{ij} := \frac{b_{ij}}{q_j}, \forall ij \in E.
        \end{align*}
        For every $j \in M$, we have $\sum_{i \in N_j} x_{ij} = \sum_{i \in N_j} \frac{b_{ij}}{q_j} = 1$. For every $i \in N$, we have $\sum_{j \in M_i}x_{ij} = \sum_{j \in M_i} \frac{b_{ij}}{q_j} \geq \sum_{j \in M_i} b_{ij} = 1$. So, $x$ is a valid solution to CP\eqref{CP:NSW}. 

        Focus on any agent $i \in N$. We have
        \begin{align*}
            \sum_{j \in M_i} \min\{v_{ij}, 1\} x_{ij} = \sum_{j \in M_i} \min\{p_j, 1\} x_{ij} = \sum_{j \in M_i} q_j x_{ij} = \sum_{j \in M_i} b_{ij} = 1.
        \end{align*}
        The first equality follows from the third property in Lemma~\ref{lemma:pricing}, and the second equality follows from the first property. Hence, we have $h_i(x_i) = 1$. (The choice is unique for if $|x_i|_1 > 1$, and it is valid if $|x_i|_1 = 1$.) So, 
        \begin{align*}
            f_i(x_i) = \sum_{j \in M_i:v_{ij} > 1} x_{ij} \ln v_{ij} = \sum_{j \in M_i: p_j > 1} b_{ij} \ln v_{ij} = \sum_{j \in M_i} b_{ij} \ln \frac{v_{ij}}{q_j}.
        \end{align*}
        The first equality is by the definition of $f_i(x_i)$ and that $h_i(x_i) = 1$, the second one is by that if $x_{ij} > 0$ and $v_{ij} > 1$, then $v_{ij} = p_j > 1$ and $q_j = 1$, and the third one is by that $q_j = \min\{p_j, 1\}$ and that $b_{ij} > 0$ implies $v_{ij} = p_j$. 

        So, 
        \begin{flalign*}
            &&\frac1n\sum_{i \in N}f_i(x_i) = \frac1n\sum_{ij \in E}b_{ij}\ln\frac{v_{ij}}{q_j} = \frac1n \left(\sum_{ij\in E} b_{ij}\ln v_{ij} - \sum_{j \in M}q_j \ln q_j\right). &&\qedhere
        \end{flalign*}
    \end{proof}

\section{Analysis of Rounding Algorithm for CP\eqref{CP:NSW}: Proof of Theorem~\ref{thm:rounding}}
\label{sec:rounding}
After solving CP\eqref{CP:NSW} up to an additive error of $\ln(1+\epsilon)$, we obtain the fractional solution $x \in [0, 1]^E$. We then use the rounding algorithm described in Section~\ref{sec:rounding-framework} to output an integral allocation $\rho$ of items to agents. In this section, we show that the approximation ratio of the algorithm is $e^{1/e}$, matching the ratio given by \cite{theoretics:14642}. This proves Theorem~\ref{thm:rounding}.  As in \cite{theoretics:14642}, it is more convenient to consider the randomized version of the rounding algorithm. 

Till the end of this section, we focus on a fixed agent $i \in N$. Then, we create an unweighted NSW instance $\calI$ with $\Delta$ copies of agents $i$ (indexed by $[\Delta]$), and an EF1 allocation $\sigma$ for $\calI$, as described in Section~\ref{sec:creating-identical}. 

We shall use $M'$ to denote the set of items in $\calI$, and $M'' \subseteq M'$ to denote the $\Delta$ largest items in $M'$, according to the valuation function $v_i$ and breaking ties using the total order $\succ_i$ specified in Section~\ref{sec:creating-identical}. So, $\sigma$ allocates precisely one item in $M''$ to any agent in $[\Delta]$.

\begin{figure}
    \centering
    \includegraphics[width=\textwidth]{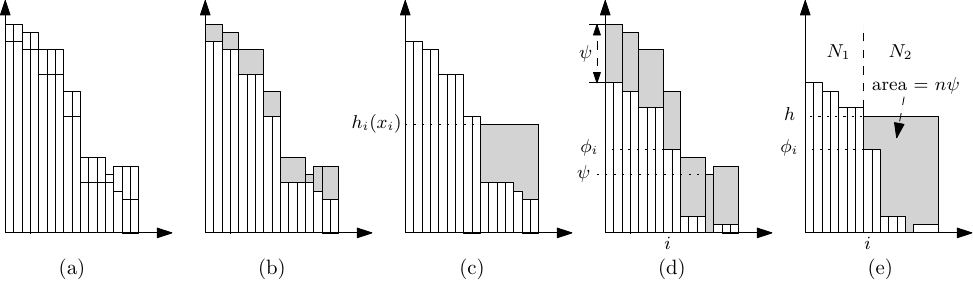}
    \caption{Liquidization Operations in Analysis of Rounding Algorithm in Section~\ref{sec:rounding} and Proof of the EF1-Gap (Theorem~\ref{thm:EF1-gap}) in Section~\ref{sec:EF1-gap}. Each rectangle represents a solid item, with height indicating its value. Gray solid polygons represent liquid items. Each column represents an agent, and the rectangles and the gray area in the column correspond the solid and liquid items allocated to the agent. All agents are identical. Figure (a) denotes the instance $\calI$ and the EF1-allocation $\sigma$ created in Section~\ref{sec:rounding}. Figure~(b) denotes the instance $\calI'$ obtained from $\calI$ after liquidization operations and the allocation $\sigma'$. Figure~(c) shows the optimum solution for $\calI'$. Figure~(d) denotes the liquidized instance for $\calI$ when we try to prove $\sigma$ is $e^{1/e}$-approximate in Section~\ref{sec:EF1-gap}. Figure~(e) denotes the optimum solution for the instance. $\phi_i$'s, $\psi, h, N_1$ and $N_2$ are depicted in Figures~(d) and (e).}
    \label{fig:liqudization}
\end{figure}

\paragraph{Creating $\calI'$ and $\sigma'$ via Liquidization of Items}
We define a \emph{liquidization} operation over $\calI$ and $\sigma$ as follows. Given an item $j$, \emph{liquidizing} item $j$ means splitting $j$ into many sufficiently-small items whose total value is $v_j$, and we call them \emph{liquid} items. Once an item is liquidized, we update the allocation $\sigma$ accordingly. The original items in $\calI$ that were not liquidized are called \emph{solid} items.

We liquidize all items in $M' \setminus M''$ in the instance $\calI$, and let $\calI'$ be the new instance. Let $\sigma'$ be the allocation $\sigma$ after liquidization. The set of solid items becomes $M''$; and $\sigma$ allocates precisely one solid item to an agent in $[\Delta]$.  Every agent $i' \in [\Delta]$ gets the same value in $\sigma'$ as he gets in $\sigma$; so the NSW value of $\sigma'$ is the same as that of $\sigma$.  Moreover, $\sigma'$ is an EF1 allocation for $\calI'$. The liquidization operation can only make the instance easier, and so the optimum NSW value of $\calI'$ is at least that of $\calI$. See Figure~\ref{fig:liqudization}(a) and \ref{fig:liqudization}(b) for an illustration of the liquidization process.

\paragraph{Comparing $\sigma'$ and Optimum Value of $\calI'$} We now analyze the optimum NSW value of $\calI'$. One can prove that in the optimum solution, every agent gets precisely one solid item: if some agent $i'$ gets at least $2$ solid items, then some agent $i''$ gets none. We can move one solid item from $i'$ to $i''$ and some liquid items from $i''$ to $i'$ without decreasing the NSW value.

If $M' = M''$, then there are no liquid items; clearly $\sigma'$ is the optimum allocation. We assume $M'' \subsetneq M'$ and thus $v_i(M' \setminus M'') > 0$. The optimum solution to the new instance can be obtained by allocating the $v_i(M' \setminus M'')$ value of liquid items using the water-filling procedure. Recall that $h_i(x_i)$ is the unique real $h$ such that $\sum_{j \in M_i}z^{i, 1}_j \min \{v_{ij}, h\} + \sum_{j \in M_i}(x_{ij} - z^{i, 1}_j) v_{ij} = h$, which is equivalent to $\frac{1}{\Delta}\sum_{j \in M''} \min\{v_{ij}, h\} + \frac{1}{\Delta}\cdot v_i(M' \setminus M'') = h$.   See Figure~\ref{fig:liqudization}(c) for an illustration of the optimum solution of $\calI'$ and $h_i(x_i)$. 


So, the logarithm of the optimum NSW value of $\calI'$ is precisely
 \begin{align*}
     &\quad \frac1\Delta \sum_{j \in M''} \ln \max\{v_{ij}, h_i(x_i)\} = \sum_{j \in M'':v_{ij}>h_i(x_i)}\frac1\Delta \ln \frac{v_{ij}}{h_i(x_i)} + \ln h_i(x_i) \\
     &= \sum_{j \in M':v_{ij}>h_i(x_i)}\frac1\Delta \ln \frac{v_{ij}}{h_i(x_i)} + \ln h_i(x_i) = \sum_{j \in M_i:v_{ij} > h_i(x_i)} x_{ij} \ln \frac{v_{ij}}{h_i(x_i)} + \ln h_i(x_i) = f_i(x_i).
 \end{align*}
 The second equality used that every $j \in M' \setminus M''$ has $v_{ij} \leq h_i(x_i)$.  By Theorem~\ref{thm:EF1-gap}, the logarithm of the NSW of $\sigma$ is at least $f_i(x_i) - \frac1e$. Therefore, $\E[\ln v_i(\rho^{-1}(i))] \geq f_i(x_i) - \frac1e$. 

\paragraph{Wrapping up the Analysis} We now consider all agents $i \in N$. By linearity of expectation, we have 
\begin{align*}
    \E\left[\sum_{i \in N}w_i \ln v_i(\rho^{-1}(i))\right] \geq \sum_{i \in N}w_i f_i(x_i) - \frac1e.    
\end{align*}
 By the convexity of exponential function, this implies 
 \begin{align*}
   \E\left[\prod_{i \in N}v_i(\rho^{-1}(i))^{w_i}\right] \geq e^{-1/e}\cdot \exp\left(\sum_{i \in N}w_i f_i(x_i)\right).
 \end{align*}

 Recall that $\sum_{i \in N}w_i f_i(x_i)$ is the value of $x$ to \eqref{CP:NSW}. This leads to a $e^{-1/e}(1+\epsilon)$-approximation for the weighted NSW problem.

\section{A Simple Proof of the $e^{1/e}$ EF1-Gap (Theorem~\ref{thm:EF1-gap})}

\label{sec:EF1-gap}

In this section, we give a simple proof of Theorem~\ref{thm:EF1-gap}. Recall that we are given an unweighted additive NSW instance with identical agents, defined by $N, n, M, m$ and $v$, and an EF1 allocation $\sigma$. Let $V_i:=v(\sigma^{-1}(i))$ for every $i \in N$ be the value allocated to $i$. Let $\psi := \min_{i \in N}V_i$ and assume $\psi > 0$ (otherwise $\opt = 0$). 
Let $\phi_i = V_i - \psi \geq 0$ for every $i \in N$.

As in Section~\ref{sec:rounding}, we shall liquidize some of the items in the NSW instance. The operations we apply here are slightly different. See Figure~\ref{fig:liqudization}(d) for the procedure.  Focus on each agent $i \in N$. Let $j$ be the largest item allocated to $i$ in the EF1 allocation. We liquidize all items in $\sigma^{-1}(i) \setminus j$. For the item $j$, we first break $j$ into two items: one with value $\phi_i$ and the other with value $v_j - \phi_i$, and then we liquidize the item of value $v_j - \phi_i$. So the total value of liquid items allocated to $i$ is precisely $v_j - \phi_i + (V_i - v_j) = \psi$.  This operation is valid as $V_i - v_j \leq \psi$, implied by that $\sigma$ is EF1. We remark that liquidizing a portion of the item $j$ is crucial for our simplified proof.  \medskip

After the liquidization operations, $\sigma$ remains EF1, and the value allocated to any agent does not change.  Moreover, the instance becomes easier, and so $\opt$ can only increase. Therefore, it is sufficient to focus on the current instance. As before, the non-liquid items are called \emph{solid} items. In the allocation $\sigma$, every agent $i$ gets a solid item of value $\phi_i$ (unless $\phi_i = 0$), and liquid items with total value $\psi$. So the NSW value of $\sigma$ is 
\begin{align*}
    \left(\prod_{i \in N}(\phi_i + \psi)\right)^{1/n}.
\end{align*}

As before, we can prove that every agent $i$ gets at most one solid item in the optimum solution. Therefore, the optimum allocation is obtained using the water-filling method as before. Let $h > 0$ be the unique real such that $\frac1{n}\sum_{i \in N} \min\{\phi_i, h\} + \psi = h$. We have 
\begin{align*}
    \opt = \left(\prod_{i \in N}\max\{\phi_i, h\}\right)^{1/n}.
\end{align*}
See Figure~\ref{fig:liqudization}(e) for an illustration of the optimum solution. 

\paragraph{Comparing $\opt$ and the Value of $\sigma$} Let $N_1 = \{i \in N: \phi_i > h\}$ and $N_2 := N \setminus N_1 = \{i \in N: \phi_i \leq h\}$.  Let $n_1 = |N_1|$ and $n_2 = |N_2| = n - n_1$. So $h = \frac1{n}\sum_{i \in N} \min\{\phi_i, h\} + \psi  = \frac1{n}  \left(\sum_{i \in N_1} h + \sum_{i \in N_2}\phi_i\right) + \psi $, which is equivalent to $n_2 h = n\psi  + \sum_{i \in N_2} \phi_i$. 
    
    We take logarithm of $\opt$ and the value of $\sigma$, and consider the difference scaled by $n$:
    \begin{align*}
        \sum_{i \in N} \ln \max\{\phi_i, h\} - \sum_{i \in N} \ln(\phi_i + \psi) = 
        \left(\sum_{i \in N_1} \ln \frac{\phi_i}{\phi_i + \psi} + \sum_{i \in N_2}\ln\frac{h}{\phi_i + \psi}\right) \leq \sum_{i \in N_2}\ln \frac{h}{\phi_i + \psi}.
    \end{align*}
    
    As $\sum_{i \in N_2} \phi_i = n_2 h -n\psi $, and $\phi_i \in [0, h]$ for every $i \in N_2$, by concavity of logarithm, we have
    \begin{align*}
        \sum_{i \in N_2}\ln (\phi_i + \psi) &\geq \frac{n_2 h - n\psi }{h} \cdot \ln (h + \psi ) + \frac{n\psi }{h}\ln \psi  \\
        &\geq \Big(n_2 - \frac{n\psi }{h}\Big) \ln h + \frac{n\psi }{h}\ln \psi  = n_2\ln h - \frac{n \psi }{h} \ln \frac{h}{\psi }.
    \end{align*}
    Therefore, 
    \begin{align*}
         \sum_{i \in N_2} \ln \frac{h}{\phi_i + \psi}\leq \frac{n\psi }{h} \ln \frac{ h}{\psi } \leq \frac{n}e.
    \end{align*}
    The second inequality used that $\frac1x \ln x$ for $x \geq 1$ obtains its maximum value $\frac1e$ at $x = e$.  This proves that 
    \begin{align*}
        \sum_{i \in N} \ln \max\{\phi_i, h\} - \sum_{i \in N} \ln(\phi_i + \psi) \leq \frac{n}e.
    \end{align*}
    So, $\left(\prod_{i \in N} \max \{\phi_i, h\}\right)^{1/n} \leq e^{1/e} \cdot \left(\prod_{i \in N} (\phi_i + \psi)\right)^{1/n}$.  This finishes the proof of Theorem~\ref{thm:EF1-gap}.

\section{Application to Scheduling Problems}
\label{sec:scheduling}
In this section, we show that the convex program techniques can be applied to two unrelated machine scheduling problems. 

\subsection{Scheduling to Minimize $\sum_{i \in M}\theta(\text{load}_i)$}
\label{sec:Lk-norm}

In this section, we consider the unrelated machine scheduling problem to minimize $\sum_{i \in M}\theta(\text{load}_i)$ as defined in Section~\ref{sec:problem-definitions}, where $\theta:\R_{\geq 0} \to \R_{\geq 0}$ is an increasing convex function with $\theta(0) = 0$. It contains the problems of minimizing $L_q$ norm as special cases.  

We need to assume the first and second-order derivatives $\theta', \theta'': \R_{\geq 0} \to \R_{\geq 0}$ exist and are continuous. By the monotonicity and convexity of $\theta$, we have $\theta'(t), \theta''(t) \geq 0$ for every $t > 0$.

As before, we fix a machine $i \in M$, and an allocation $x_i \in [0, 1]^J$.  If $|x_i|_1>1$, we define $h_i(x_i)$ to be the unique real $h > 0$ such that 
\begin{align*}
    \sum_{j \in J} \min\{p_{ij}, h\} x_{ij} = h.
\end{align*}
If $|x_i|_1\leq1$, we define $h_i(x_i) = 0$, which also satisfies the above equality. 

Then, we define 
\begin{align*}
    f_i(x_i):=\sum_{j\in J: p_{ij} > {h_i(x_i)}} x_{ij} \big(\theta(p_{ij})-\theta(h_i(x_i))\big) + \theta(h_i(x_i)).
\end{align*}

For every $x_i\in[0,1]^{J}$ and $h\ge 0$ we define $$g_i(x_i, h) := \sum_{j\in J:p_{ij}>h}x_{ij}(\theta(p_{ij}) - \theta(h)) + \theta(h) + \theta'(h)\left(\sum_{j\in J}\min\{p_{ij}, h\}x_{ij}-h \right).$$

\begin{lemma}
    For any $i \in N$, $x_i \in [0, 1]^{J}$, we have
    \begin{align*}
        f_i(x_i)=\max_{h > 0}g_i(x_i, h) = g_i(x_i, h_i(x_i)).
    \end{align*}
\end{lemma}

\begin{proof}
    Throughout the proof, we fix $x_i$. We calculate $\frac{\partial g_i}{\partial h}$. We assume $h > 0$ and $h \neq p_{ij}$ for any $j \in J$:
    \begin{align*}
        \frac{\partial g_i}{\partial h} &= -\sum_{j \in J:p_{ij}>h}x_{ij}\theta'(h) + \theta'(h) + \theta''(h) \left(\sum_{j \in J}\min\{p_{ij}, h\}x_{ij}-h\right) + \theta'(h)\left(\sum_{j \in J:p_{ij}>h}x_{ij} - 1\right)\\
        &= \theta''(h)\left(\sum_{j \in J}\min\{p_{ij}, h\}x_{ij}-h\right).
    \end{align*}
    So $\frac{\partial g_i}{\partial h}$ is continuous over $(0, \infty)$. 
    As $\theta''(h) \geq 0$ for every $h > 0$, by our definition of $h_i(x_i)$, we have $\frac{\partial g_i}{\partial h} \geq 0$ when $h < h_i(x_i)$ and $\frac{\partial g_i}{\partial h}|_h \leq 0$ when $h > h_i(x_i)$. So, the function is maximized when $h = h_i(x_i)$. Also, when $h = h_i(x_i)$, we have $g_i(x_i, h) = f_i(x_i)$    as $\sum_{j\in J}\min\{p_{ij}, h_i(x_i)\}x_{ij} = h_i(x_i)$.
\end{proof}

\begin{coro}
    $f_i$ is a convex function over its domain.
\end{coro}

A similar compact convex program as follows:
\begin{align}
    \min \qquad \sum_{i \in M} f_i(x_i) \label{CP:theta}
\end{align}\vspace*{-15pt}
\begin{align*}
    \sum_{i \in M} x_{ij} = 1, \forall j \in J; \qquad 
    x_{ij} \geq 0, \forall i\in M, j \in J.
\end{align*}


We briefly talk about how to approximately solve the convex program in Appendix~\ref{appendix:lp-solver}. 
Once we solved the problem, we use the rounding algorithm as described in Section~\ref{sec:rounding-framework}. To analyze the approximation ratio, we focus on a fixed machine $i \in M$.  We construct a scheduling instance $\calI$ with $\Delta$ copies of machine $i$, which are indexed by $[\Delta]$, as described in Section~\ref{sec:creating-identical}. Let $J'$ be the set of $n':=|J'|$ jobs in $\calI$, where each job $j \in J'$ has processing time $p_j := p_{ij}$. Also created is a EF1 allocation $\sigma: J' \to [\Delta]$ of jobs to machines.\footnote{Although in the scheduling problem, we need to minimize loads allocated to machines, envy-freeness is still defined by interpreting jobs as goods.}

As in Section~\ref{sec:rounding}, we liquidize all jobs in $\calI$ except for the largest job allocated to each $i' \in [\Delta]$. Let $\calI'$ be the resulting scheduling instance with identical machines, and $\sigma'$ be new allocation. $\sigma'$ remains EF1, and its cost is the same as that of $\sigma$. The optimum cost of $\calI$ is at least that of $\calI'$, which is precisely $f_i(x_i)$. We prove the following theorem:


\begin{theorem}
    \label{thm:EF1-gap-Lq}
    Consider the identical machine scheduling problem to minimize $\sum_{i \in M}\theta(\text{load}_i)$. 
    Assume 
    \begin{align*}
        \alpha := \sup_{t\in(0,1); d\geq0; r\geq \frac{d}{1-t}}\frac{t\theta(r+d)+(1-t)\theta(d)}{t\theta(r)+(1-t)\theta(\frac{d}{1-t})} < \infty.
    \end{align*}
    If $\sigma$ is an EF1 allocation for the scheduling instance with identical machines, then it is $\alpha$-approximate.
\end{theorem}

\begin{proof}
    The notation used in the proof is independent of that used elsewhere. We shall use $M$ to denote the set of identical machines, $J$ to denote the set of jobs, and $p_j > 0$ to denote the size of a job $j \in J$.  We let $m = |M|$. 
    
    We use similar notations and perform the same liquidization operations as in Section~\ref{sec:EF1-gap}. $P_i := p(\sigma^{-1}(i))$ for every machine $i \in M$. $\psi := \min'_{i \in M} P_i$, and $\phi_i := P_i - \psi$. After the liquidization operations, each machine $i$ gets a \emph{solid job} of size $\phi_i$, and \emph{liquid jobs} of total size $\psi$. $\sigma$ remains EF1, and its cost does not change. Moreover, the instance only becomes easier and thus the optimum cost only decreases. So, it suffices to focus on the instance obtained after the liquidization operations.  The cost of $\sigma$ is $\sum_{i \in M}\theta(\phi_i + \psi)$. 

    We can assume $\psi > 0$; otherwise $\sigma$ is optimum and the theorem holds trivially. Let $h$ be the unique real such that $\frac{1}{m}\sum_{i \in M}\min\{\phi_i, h\} + \psi = h$. We have $\opt = \sum_{i \in M}\theta(\max\{\phi_i, h\})$. \medskip

    Similarly, we let $M_1 = \{i \in M: \phi_i > h\}$ and $M_2 = M \setminus M_1 = \{i \in M:\phi_i \leq h\}$. Let $m_1 = |M_1|$ and $m_2 = |M_2| = m - m_1$. So $h = \frac1{n}\sum_{i \in M} \min\{\phi_i, h\} + \psi  = \frac1{m}  \left(\sum_{i \in M_1} h + \sum_{i \in M_2}\phi_i\right) + \psi$, which is equivalent to $m_2 h = m\psi  + \sum_{i \in M_2} \phi_i$. 

    By convexity of $\theta$, and that $\phi_i \in [0, h]$ for every $i \in M_2$, we have 
    \begin{align*}
        \sum_{i\in M_2}\theta(\phi_i+\psi)\leq \frac{m_2h-m\psi}{h} \theta(h+\psi)+\frac{m\psi}{h} \theta(\psi) = \left(m_2 - \frac{m\psi}{h}\right)\cdot\theta(h+\psi) + \frac{m\psi}{h}\cdot\theta(\psi).
    \end{align*}

    Let $t=1-\frac\psi h$. Focus on a machine $i \in M_1$. We have $\phi_i> h = \psi/(1-t)$. By the definition of $\alpha$ and setting $d = \psi$, $r = \phi_i$, we have
    \begin{align*}
        t\theta(\phi_i+\psi)+(1-t)\theta(\psi) \leq \alpha \cdot \left(t\theta(\phi_i) + (1-t) \theta(h) \right).
    \end{align*}
    Setting $d = \psi$ and $r = h$, we have
    \begin{align*}
        t\theta(h+\psi)+(1-t)\theta(\psi) \leq \alpha\cdot\left( t\theta(h) + (1-t)\theta(h) \right) = \alpha \cdot \theta(h).
    \end{align*}

    Therefore
    
    \begin{flalign*}
       &&  \sum_{i\in M}\theta(\phi_i + \psi) &= \sum_{i\in M_1}\theta(\phi_i+\psi) + \sum_{i\in M_2}\theta(\phi_i+\psi) \\
       && &\leq \sum_{i\in M_1}\theta(\phi_i+\psi) +  \left(m_2 - \frac{m\psi}{h}\right)\cdot\theta(h+\psi) + \frac{m\psi}{h}\cdot\theta(\psi)\\
       && &= \sum_{i\in M_1}\theta(\phi_i+\psi) + (mt-m_1)\theta(h+\psi) + m(1-t)\theta(\psi) \\
       && &= \sum_{i\in M_1}\frac{1}{t}\big(t\theta(\phi_i+\psi)+(1-t)\theta(\psi)\big) + \left(m - \frac{m_1}{t}\right)\big( t\theta(h+\psi) + (1-t)\theta(\psi) \big) \\
       && &\leq \sum_{i\in M_1}\frac{\alpha}{t}\left( t\theta(\phi_i)+(1-t)\theta(h) \right) + \alpha\left(m - \frac{m_1}{t}\right)\theta(h) \\
       && &= \alpha\sum_{i \in M_1}\theta(\phi_i) + \alpha \left(\frac{m_1 (1-t)}{t} + m - \frac{m_1}{t}\right)\theta(h) \quad = \quad\alpha\sum_{i \in M_1}\theta(\phi_i) + \alpha m_2\theta(h) \\
        && &= \alpha\left(\sum_{i\in M_1}\theta(\phi_i)+\sum_{i\in M_2}\theta(h)\right) \quad = \quad \alpha\sum_{i \in M}\theta(\max\{\phi_i, h\}). &&\qedhere
    \end{flalign*}
\end{proof}

When the objective is the $L_k$ norm of machine loads, we obtain an $\alpha^{1/k}$-approximation algorithm, where
\begin{align*}
    \alpha &= \sup_{t\in(0,1); d\geq0; r\geq \frac{d}{1-t}}\frac{t(r+d)^k+(1-t)d^k}{tr^k+(1-t)(\frac{d}{1-t})^k} \\ 
    &= \sup_{t\in(0,1); y\in (0,1]}\frac{t(1+y(1-t))^k + (1-t)(y(1-t))^k}{t+(1-t)y^k}, \qquad \text{by setting } y = \frac{d}{(1-t)r}
\end{align*}
which is the same as the approximation ratio given by \cite{soda/ImLi23}.

\subsection{Weighted Completion Time Minimization with Uniform Smith Ratios}
In this section, we consider the unrelated machine weighted completion time minimization problem when jobs have uniform Smith ratios, as defined in Section~\ref{sec:problem-definitions}. Recall that the cost of an allocation $\rho:J \to M$ is 
\begin{align*}
    \frac12\sum_{i \in M}\left(\sum_{j \in \rho^{-1}(i)}p_{ij}^2 + \left(\sum_{j \in \rho^{-1}(i)}p_{ij}\right)^2\right).
\end{align*}

So the problem is closely related to the problem of minimizing the $L_2$ norm of machine loads, with the main difference being that we also incur a cost of $\sum_{j \in \rho^{-1}(i)}p_{ij}^2$ on machine $i$. The constant $\frac12$ does not change the problem. Also, we do not take a square root on the objective in the problem, and this will square the approximation ratio. 

We define $h_i(x_i)$ as in Section~\ref{sec:Lk-norm}. Treating the function $\theta$ as $\theta(t) = t^2$, we define
\begin{align*}
    f_i(x_i) := \sum_{j \in J: p_{ij} > h_i(x_i)}x_{ij}(p_{ij}^2 - h_i(x_i)^2) + h_i(x_i)^2.
\end{align*}
Then, the convex program becomes
        \begin{align}
            \min \qquad \frac12\sum_{i \in M} \left(f_i(x_i) + \sum_{j\in J}x_{ij}p_{ij}^2\right)
        \end{align}\vspace*{-15pt}
        \begin{align*}
            \sum_{i \in M} x_{ij} = 1,\forall j \in J; \qquad
            x_{ij} \geq 0,\forall i \in M, j \in J.
        \end{align*}


Once we solve the convex program to obtain a fractional allocation $x \in [0, 1]^{M \times J}$, we use the rounding algorithm of \cite{shmoys1993approximation} (described in Section~\ref{sec:rounding-framework}) to obtain an integral allocation $\rho:J \to M$ of jobs to machines. In the analysis, we fix a machine $i \in M$, and create a scheduling instance $\calI$ with $\Delta$ copies of $i$, and an EF1 allocation $\sigma$ of $\calI$.  Now all the machines are identical, the term $\frac12\sum_{i' \in [\Delta], j \in \rho^{-1}(j)} p_{ij}^2$ in the objective becomes a fixed term which is independent of the allocation. However, it will affect the approximation ratio and thus can not be removed from the objective. 

As before, for every machine $i' \in [\Delta]$, we liquidize every job in $\sigma^{-1}(i')$ except the largest one. After the operation, $\sigma$ remains EF1. One minor point is that the liquidization operations will reduce the fixed term in both the cost of $\sigma$, and in the optimum cost.  However, reducing the fixed term can only make the approximation ratio worse and thus this is not an issue.  Using the same analysis as before, it remains to prove the following theorem:

\begin{theorem}
    Consider a scheduling instance with identical machines. Then any EF1 allocation  $\sigma$ is $\alpha$-approximate, where $\alpha = \frac{\sqrt{2}+1}{2}$.
\end{theorem}

\begin{proof}
    Again, we use $M, m, J, n$ and $(p_j)_{j \in J}$ to denote the scheduling instance as in the proof of Theorem~\ref{thm:EF1-gap-Lq}. The notations in the proof are independent of that used elsewhere. We perform the same liquidization operations. Again, the operations can only decrease the fixed term and make the approximation ratio of $\sigma$ worse.  

    Following the same notations as in the proof of Theorem~\ref{thm:EF1-gap-Lq}, we have that the cost of $\sigma$ is 
    \begin{align*}
        \cost(\sigma):=\frac12\sum_{i \in M}\big((\phi_i + \psi)^2 + \phi_i^2\big).    
    \end{align*}
     We assume $\psi > 0$. Let $h$ be the unique real such that $\frac1m \sum_{i \in M}\min\{\phi_i, h\} + \psi = h$. We have 
     \begin{align*}
        \opt = \frac12\sum_{i\in M}\big(\max\{\phi_i, h\}^2 + \phi_i^2\big).         
     \end{align*}

    Similarly, we let $M_1 = \{i \in M: \phi_i > h\}$ and $M_2 = M \setminus M_1 = \{i \in M:\phi_i \leq h\}$. Let $m_1 = |M_1|$ and $m_2 = |M_2| = m - m_1$. So $h = \frac1{n}\sum_{i \in M} \min\{\phi_i, h\} + \psi  = \frac1{m}  \left(\sum_{i \in M_1} h + \sum_{i \in M_2}\phi_i\right) + \psi$, which is equivalent to $m_2 h = m\psi  + \sum_{i \in M_2} \phi_i$.

    Then, we consider
    \begin{align*}
        2(\cost(\sigma) - \alpha \cdot \opt) &= \sum_{i \in M}((\phi_i + \psi)^2 + \phi_i^2) - \alpha\sum_{i \in M_1}2\phi_i^2 - \alpha\sum_{i \in M_2}(h^2 + \phi_i^2)\\
        &= \sum_{i \in M_1}((\phi_i + \psi)^2 - (2\alpha-1)\phi_i^2) + \sum_{i \in M_2} ((\phi_i + \psi)^2- (\alpha-1)\phi_i^2)  - \alpha  m_2 h^2.
    \end{align*}

    For $i \in M_1$, we have 
    \begin{align*}
        (\phi_i + \psi)^2 - (2\alpha-1) \phi_i^2 = -2(\alpha-1)\phi_i^2 + 2\psi\phi_i + \psi^2 \leq \psi^2 + \frac{(2\psi)^2}{4\cdot 2(\alpha-1)} = \frac{2\alpha-1}{2(\alpha-1)}\cdot \psi^2.    
    \end{align*}
    Since the function $(x+\psi)^2-(\alpha - 1) x^2$ is convex, $\sum_{i \in M_2}\phi_i = m_2h - m\psi$ and $\phi_i \in [0, h]$ for every $i \in M_2$, we have 
    \begin{align*}
        \sum_{i\in M_2}\left((\phi_i+\psi)^2-(\alpha-1)\phi_i^2\right) \leq \left(m_2 - \frac{m\psi}{h}\right)((h+\psi)^2-(\alpha - 1) h^2)+\frac{m\psi}{h}\psi^2.
    \end{align*}
    
    Therefore,
    \begin{align*}
        2(\cost(\sigma) - \alpha \cdot \opt) &\leq \frac{2\alpha-1}{2(\alpha-1)} \cdot m_1 \psi^2 + \left(m_2 - \frac{m\psi}{h}\right)((h+\psi)^2-(\alpha - 1) h^2)+\frac{m\psi}{h}\psi^2 - \alpha m_2h^2.
    \end{align*}

    Now, we let $a = \frac\psi h \in [0, 1]$ and $b = \frac{m_1}{m} \in [0, 1]$. Notice that $m_2 h \geq m\psi$, which implies $ a = \frac{\psi}{h} \leq \frac{m_2}{m} = 1 - b$. Then, we have 
    \begin{align}
        \frac{2(\cost(\sigma) - \alpha \cdot \opt)}{mh^2} &\leq
        \frac{2\alpha-1}{2(\alpha-1)} a^2 b + \left(1-b - a\right)((1+a)^2-(\alpha - 1))+a^3 - \alpha (1-b). \label{equ:bound}
    \end{align}

    We calculate the maximum of the right side of \eqref{equ:bound} subject to $a \geq 0, b \geq 0, a + b \leq 1$. The quantity is a linear function of $b$ for fixed $a$. Therefore, the maximum is achieved when $b = 0$ or $b = 1 - a$.

When $b = 0$, the right side of \eqref{equ:bound} becomes
\begin{align*}
    &\quad \left(1 - a\right)((1+a)^2-(\alpha - 1))+a^3 - \alpha = (2- a^2) - (2-a)\alpha = -a^2 + \alpha a - (2(\alpha-1)) \\
    &\leq -2(\alpha-1) + \frac{\alpha^2}{4} = \frac{3 + 2\sqrt{2}}{16} - (\sqrt{2}-1) = \frac{19 - 14\sqrt{2}}{16} < 0.
\end{align*}

When $b = 1 - a$, the right side of \eqref{equ:bound} becomes
\begin{align*}
    &\quad \frac{2\alpha-1}{2(\alpha-1)} a^2 (1-a) +a^3 - \alpha a = -\frac{1}{2(\alpha-1)}\cdot a^3 + \frac{2\alpha-1}{2(\alpha-1)}\cdot a^2 - \alpha a \\
    &= \frac{a}{2(\alpha-1)}\big(-a^2 + (2\alpha-1)a - 2\alpha(\alpha-1)\big)
\end{align*}
Then,
\begin{align*}
    -a^2 + (2\alpha-1)a - 2\alpha(\alpha-1) &\leq -2\alpha(\alpha-1) + \frac{(2\alpha-1)^2}{4} = -\alpha^2 +\alpha + \frac14 = 0,
\end{align*}
as $\alpha = \frac{\sqrt{2}+1}{2}$ is a root of the equation $\alpha^2 - \alpha - \frac14 = 0$.

Therefore, the right side of \eqref{equ:bound} is at most $0$, which implies $\cost(\sigma) - \alpha \cdot \opt \leq 0$. This finishes the proof of the theorem. 
\end{proof}

\bibliographystyle{plain}
\bibliography{reflist}

\appendix
\section{Solving Convex Programs CP\eqref{CP:theta} Approximately}

\label{appendix:lp-solver}

Convex program can be solved within an additive error of $\epsilon$ for any given $\epsilon > 0$ under some mild requirements. For our convex program CP\eqref{CP:theta}, we can convert it to a linear program using the discretization technique. Due to the existence of the function $\theta$, it is convenient for us to impose approximation on job sizes. 

For every $i\in N$, we let $\ell_i:= \min\{p_{ij} : j\in J\}$ and $r_i:=\sum_{j\in J}p_{ij}$, then we have $h_i(x_i)=\{0\}\cup [\ell_i, r_i]$. Let $H_i = \{0\}\cup \{\ell_i(1+\epsilon)^{t}\le r_i : t\in\Z_{\geq 0}\}$ and we use $\bar{f}_i(x_i) := \max_{h\in H_i}g(x_i, h)$ as an approximation of $f_i(x_i)$. Then, we can use functions $\bar f_i$'s to replace $f_i$'s in CP\eqref{CP:theta}, and the new convex program can be explicitly formulated as a polynomial-size linear program. 

We define $f_i\left(x_i; \frac{1}{1+\epsilon}\right)$ in the same way as we define $f_i(x_i)$, except that we use the processing times $\frac{p_{ij}}{1+\epsilon}$ for all $j \in J$. The main lemma we prove is
\begin{lemma}
    $f_i\left(x_i; \frac{1}{1+\epsilon}\right) \leq \bar{f}_i(x_i) \leq f_i(x_i)$ for every $x_i \in [0, 1]^J$.
\end{lemma}
    \begin{proof}
        $\bar{f}_i(x_i) \leq f_i(x_i)$ holds trivially. So we only need to prove the first inequality. For notational convenience, we use $h_i$ to denote $h_i(x_i)$. Let $\bar h$ be the largest number in $H_i$ with $\bar h \leq h_i$. Then, we have $\frac{h_i}{1+\epsilon} < \bar h \leq h_i$. 
        \begin{align*}
            f_i\left(x_i; \frac{1}{1+\epsilon}\right)&=\sum_{j\in J: p_{ij} > {h_i}} x_{ij} \left(\theta\Big(\frac{p_{ij}}{1+\epsilon}\Big)-\theta\Big(\frac{h_i}{1+\epsilon}\Big)\right) + \theta\Big(\frac{h_i}{1+\epsilon}\Big)\\
            &\leq \sum_{j\in J:p_{ij}>h}x_{ij}\big(\theta(p_{ij}) - \theta(\bar h)\big) + \theta(\bar h) + \theta'(\bar h)\left(\sum_{j\in J}\min\{p_{ij}, \bar h\}x_{ij}-\bar h \right)\\
            &= g_i(x_i, \bar h) \leq \bar f_i(x_i).
        \end{align*}

       We explain the inequality. $\theta\Big(\frac{p_{ij}}{1+\epsilon}\Big)-\theta\Big(\frac{h_i}{1+\epsilon}\Big) \leq \theta(p_{ij}) - \theta(h_i) \leq \theta(p_{ij}) - \theta(\bar h)$. This holds as $\theta$ is monotone and convex. Also, $\theta\Big(\frac{h_i}{1+\epsilon}\Big) \leq \theta(\bar h)$ as $\frac{h_i}{1+\epsilon} < \bar h$. Finally, 
       $\theta'(\bar h)\left(\sum_{j\in J}\min\{p_{ij}, \bar h\}x_{ij}-\bar h \right) \geq 0$ as $\theta'(\bar h) \geq 0$ and $\bar h \leq h_i$. 
    \end{proof}

    For functions $\theta$ with moderate growth (including the case where $\theta(x) = x^k$ for a constant $k$), the lemma suffices to solve CP\eqref{CP:theta} up to arbitrary precision.

\end{document}